\documentclass[review]{elsarticle}

\label{packages}
\usepackage{amsfonts}
\usepackage{mathrsfs}
\usepackage{mathtools}
\usepackage{verbatim}
\usepackage{upgreek}
\usepackage{color}
\usepackage[usenames,dvipsnames]{xcolor}
\usepackage{caption}
\usepackage{subfig}
\usepackage{multirow}
\usepackage{mdwmath}
\usepackage{mdwtab}
\usepackage{amssymb}
\usepackage{amsmath}
\usepackage{amsthm}
\usepackage{epstopdf}
\usepackage{scalefnt}
\usepackage{booktabs}
\usepackage[disable]{todonotes}

\usepackage{hyperref}
\hypersetup{colorlinks=true,linkcolor=red}
\usepackage{breqn}
\usepackage[utf8]{inputenc}
\usepackage{mathrsfs}
\usetikzlibrary{arrows}
\usepackage{pgfplots}
\usepackage{graphicx}
\usepackage{overpic}

\usepackage[ruled,noend,slide]{algorithm2e}

\newcommand\ineqno{\hfill\refstepcounter{equation}\text{(\theequation)}}

\newcommand{\bzero}{\mbox{\boldmath{$0$}}}

\label{English Chars}
\newcommand{\bA}{{\bf A}}
\newcommand{\ba}{{\bf a}}
\newcommand{\bB}{{\bf B}}
\newcommand{\bb}{{\bf b}}

\newcommand{\bE}{{\bf E}}
\newcommand{\be}{{\bf e}}
\newcommand{\bF}{{\bf F}}
\newcommand{\ff}{{\bf f}}
\newcommand{\bG}{{\bf G}}

\newcommand{\bH}{{\bf H}}

\newcommand{\bI}{{\bf I}}

\newcommand{\bk}{{\bf k}}
\newcommand{\bL}{{\bf L}}

\newcommand{\bM}{{\bf M}}

\newcommand{\bp}{{\bf p}}
\newcommand{\bQ}{{\bf Q}}

\newcommand{\br}{{\bf r}}

\newcommand{\bV}{{\bf V}}

\newcommand{\bX}{{\bf X}}
\newcommand{\bx}{{\bf x}}

\newcommand{\by}{{\bf y}}

\newcommand{\bz}{{\bf z}}

\newcommand{\Diag}{\mbox{\boldmath\bf Diag}\, }

\newcommand{\Tr}{\mbox{\rm Tr}\, }

\label{Roman Chars}

\newcommand{\bmu}{\mbox{\boldmath{$\mu$}}}

\newcommand{\bLambda}{\mbox{\boldmath{$\Lambda$}}}
\newcommand{\bgamma}{\mbox{\boldmath{$\gamma$}}}

\label{Mycommands}

\label{Theorems}
\theoremstyle{definition}
\newtheorem{theorem}{Theorem}[section]

\newtheorem{lemma}[theorem]{Lemma}

\theoremstyle{remark}

\newtheorem{remark}{Remark}

\journal{Journal of \LaTeX\ Templates}

\begin{document}
	
	\begin{frontmatter}
		
		\title{Array Resource Allocation for Radar and Communication Integration Network}
		
		\author{Zhenkai Zhang\fnref{myfootnote}}
		\fntext[myfootnote]{Zhenkai Zhang is with the department of electronic and information, Jiangsu University of Science and Technology, and is currently a visiting professor at the Department of Electrical and Computer Engineering, University of Calgary, Calgary, AB T2N 1N4, Canada; Email: zhangzhenkai@just.edu.cn. Phone: +1(403) 903-7929.}
		
		\author{Hamid Esmaeili Najafabadi*\fnref{myfootnote2}}
		\fntext[myfootnote2]{Hamid Esmaeili Najafabadi is currently a PostDoc associate at the Department of Electrical and Computer Engineering, University of Calgary, Calgary, AB T2N 1N4, Canada; Email: hamid.esmaeili@gmail.com. Phone: +1(587) 322-3311.}
		\author{Henry Leung\fnref{myfootnote3}}
		\fntext[myfootnote3]{Prof. Henry Leung is with the Department of Electrical and Computer Engineering, University of Calgary, Calgary, AB T2N 1N4, Canada; Email: leungh@ucalgary.ca.\\
		This work is supported by the National Natural Science Fund of China grants 61871203 and 61701416, China Postdoctoral Science Foundation grant 2016M592334  and Qing Lan Project of Jiangsu Province.\\
		© 2020. This manuscript version is made available under the CC-BY-NC-ND 4.0 license http://creativecommons.org/licenses/by-nc-nd/4.0/
		}
		
		
		\cortext[mycorrespondingauthor]{Corresponding author}
		%
		
		\begin{abstract}
	A radar and communication integration (RCI) system has great flexibility in allocating antenna resources to guarantee both radar and communication performance. This paper considers the array allocation problems for multiple target localization and multiple platforms communication in an RCI network. The objective of array allocation is to maximize the communication capacity for each channel and to minimize the localization error for each target. In this paper, we firstly build a localization and communication model for array allocation in an RCI network. Minorization maximization (MM) is then applied to create surrogate functions for multiple objective optimization problems. The projected gradient descent (PGD) method is further employed to solve two array allocation problems with and without a certain communication capacity constraint. Computer simulations are conducted to evaluate the performance of the proposed algorithms. The results show that the proposed algorithms have improved localization and communication performance after efficiently allocating the array resource in the RCI network.
		\end{abstract}
		
		
	\end{frontmatter}
	

	\section{Introduction}
	{Considerable} effort has been put over the past years in studying radio frequency integration technology to share the hardware and software resources \cite{Moghaddasi2016, Pastina2018}.
	Spectrum sharing technology between multiple-input multiple-out (MIMO) radar and communication is studied in \cite{Singh2018}, where an optimization algorithm is presented to design the integrated transceivers that can maximize the radar detection probability and guarantee the communication quality.
	A spectrum sharing algorithm is proposed in \cite{Zhang2019} that incorporates communication information into radar waveforms. When the radar and communication systems are operated over the same frequency band,  orthogonal frequency division multiplexing (OFDM) waveform design based on power minimization under mutual information constraints is considered in \cite{Shi2018}. 
	Two waveform designs are proposed for an OFDM integrated radar and communication system \cite{Liu2017}. A dual-function radar-communication system is proposed by using the sidelobe manipulation concept \cite{Nusenu2019}, where the communication signals are transmitted at the null radiation direction of the radar's main beam. The waveform is designed to minimize the multi-user interference by developing an appropriate beam pattern \cite{Liu2018}. Using a similar approach, communication symbols are embedded into the radar waveform by introducing a weighted coefficient to make a balance between the communication performance and radar sidelobe \cite{Gu2018}. An integrated vehicular radar-communication system at 60 GHz is developed in \cite{Kumari2018} based on the auto-correlation property at zero-Doppler. A joint radar-communication system is designed based on time modulated array in \cite{Shan2018} according to the civil and military requirements.
	These works on radar and communication integration (RCI) mainly focus on addressing spectrum sharing \cite{Singh2018}-\cite{Zhang2019}, waveform optimization \cite{Shi2018}-\cite{Gu2018} and system design \cite{Kumari2018}-\cite{Shan2018}. The problem of resource allocation for the integration network is rarely considered.
	
	Meanwhile, networked radar systems have been shown to offer better performance for target tracking or localization. 
	A joint antenna selection and power allocation for MIMO radar networks based on convex optimization method is presented in \cite{Ma2014}. In the same MIMO radar network, power combined with bandwidth and beam is optimized for the best radar performance \cite{Garcia2014} and \cite{Yan2016}. References  \cite{Deligiannis2017, Chen2015} and \cite{Alirezaei2015} present the power allocation schemes for target detection, target tracking, and target classification in a radar network, respectively. An adaptive radar receivers placement approach is proposed in \cite{BenKilani2018} to maximize the signal-to-interference-plus ratio for all channels. References \cite{Yan2017, Xie2018}  develop power allocation methods for target tracking in a radar network, which employ optimization methods to allocate power resources.  Using the maximum block improvement method in cellular networks and radar systems, an optimization framework is developed for resource allocation in \cite{Aubry2018}.
	Resource optimization is solved in \cite{Zhou2019} for the wireless-powered integrated radar and communication system subject to the performance constraints. A transmit antenna selection method for iterative receivers is presented in \cite{Kim2007}, and a joint transmitter and receiver antenna selection method is investigated \cite{Coskun2011, Ju2010, Yang2015}. The transmit power and the number of active antennas are jointly optimized to get the highest energy efficiency \cite{Jiang2012}. The performance of actual antenna systems using the antenna selection methods is evaluated and examined in \cite{Papamichael2011, Yilmaz2014, Fuchs2016}.
	
	Although the above-mentioned radar resource management  \cite{Ma2014}-\cite{Zhou2019} and antenna selection  \cite{Kim2007}-\cite{Fuchs2016} methods are informative, they are developed for a single radar or communication scenario. For complex RCI networks, the advantages of array resource allocation have not been realized.  
	
	This paper investigates the array allocation problem through mathematical derivations and computer simulations for multiple channel communication and multiple target localization. The main contributions of this paper are summarized as follows:
	\begin{enumerate}
		\item 
		An array allocation model is developed for multiple channel communication and multiple target localization in an RCI network, which is then converted to a linear model.
		\item  
		The MM optimization method is employed to tackle the multi-objective problem, i.e., maximizing the communication capacity and minimizing the tracking errors with the total array resource constraint.
		After designing the surrogate function based on the MM, it is converted to be a single objective optimization problem for multiple target localization. It is then solved by a joint method of MM and PGD.    
	\end{enumerate}

	\textit{Notations}: Bold uppercase (e.g., \textbf{H}) and lowercase  (e.g., \textbf{b}) letters represent the matrices
	and vectors, respectively. The notations $Tr(\cdot)$, $(\cdot)^T$ and $(\cdot)^H$ stand for trace, transpose and Hermitian of their argument, respectively. $\|\cdot\|_2$ denotes 
	the $l_2$ norm of a vector. $\bI_{N}$ denotes the identity matrix of the size $N\times N$, while $\mathbf{0}_{N}$  and $\mathbf{1}_{N}$ stands for  vectors of the size $ N $ with all its elements equal to zero and one, respectively. $\mathbf{1}_{M\times N}$ denotes a matrix of $N\times N$ with all its elements equal to one. $vec(\bA)$ denotes the column vector of matrix \bA. The sets of $M\times N$ integer matrices and the integer space are denoted by $\mathbb{N}^{M\times N}$ and $\mathbb{N}$. The symbols $\otimes$, $\circ$ and $ \oslash $ represent the Kronecker and Hadamard Products and Hadamard division, respectively. The use of $\circ$ together with a function indicates element-wise (a.k.a Hadamard) operation, e.g. $ \bx^{\circ 2} $ denotes element-wise square.
	Finally, the gradient of $f$ at $x$ is denoted by $\nabla f(x)$. 
	
	This paper is organized as follows. Section \ref{sec2} describes the system model for the RCI. Section \ref{sec3}  presents the array allocation problems using two different approaches after the brief introduction of MM theory. There, a hard-to-tackle multi-objective problem is converted to a simpler single-objective by applying MM and is then solved by the PGD.  The performance of the proposed algorithms is evaluated in Section \ref{sec4} through computer simulations. Section \ref{sec5} gives the conclusion of this paper.

	\section{System model}\label{sec2}
	
	\begin{figure}
		\centering
		\includegraphics[width=0.8\linewidth]{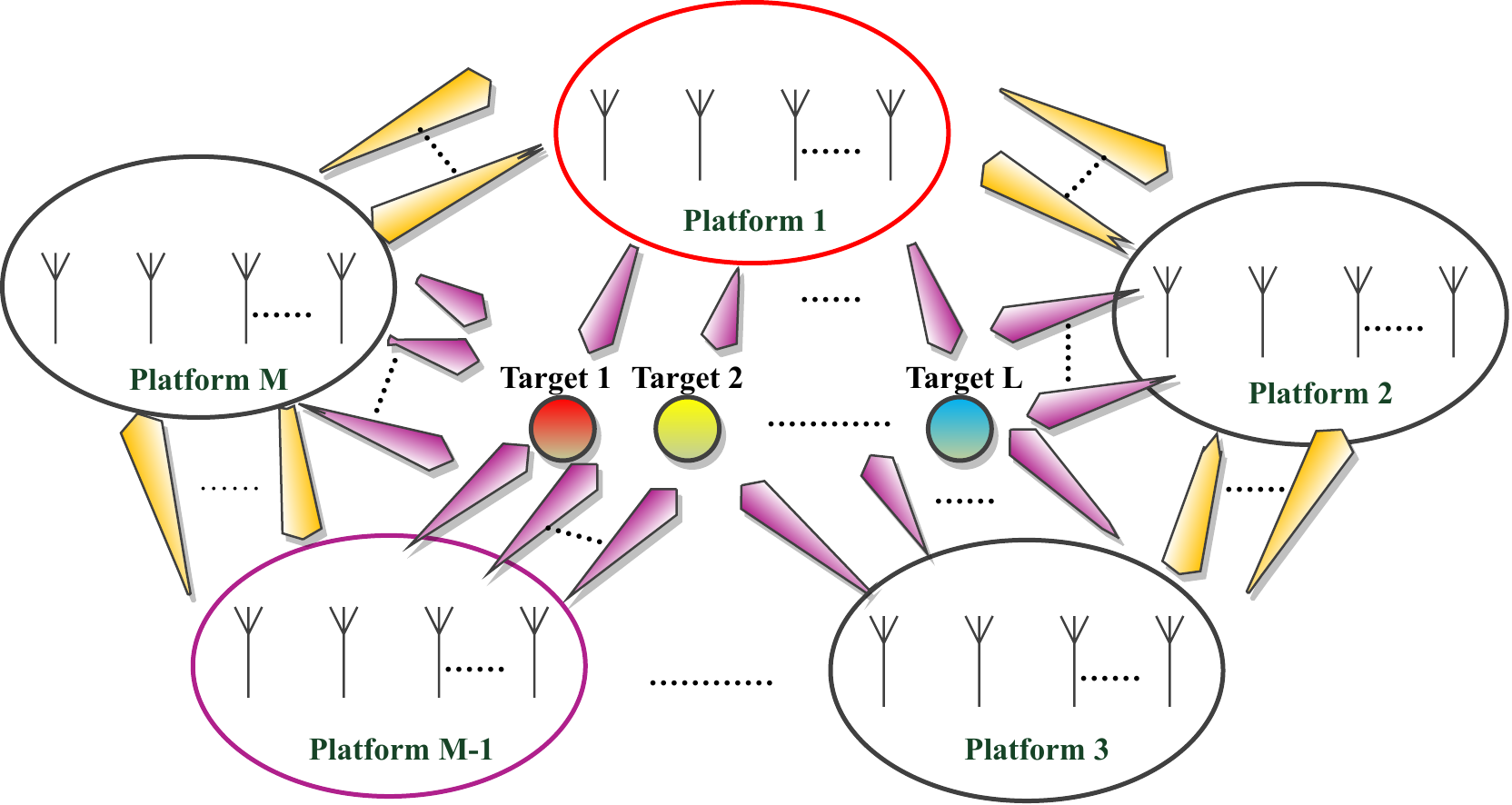}
		\caption{ Radar and communication integration network}
		\label{fig1}
	\end{figure}
	We consider the RCI network, with $M$ array platforms and $N$ targets for localization, as depicted in Fig. \ref{fig1}.
	The array resource on every platform is desired to be divided into $N+(M-1)$ parts, which are used for localization and communication, respectively. The tasks for every platform include the communication with $M-1$ platforms, and localization for $N$ targets. $\bA$ denotes the array resource allocation result for all the platforms, which will be solved by our proposed method.  
	The array resource allocation matrix $\bA$ can be written as
	\begin{equation}
	\bA=\\
	\left[
	\begin{matrix}
	a_{1,1}&\cdots& a_{1,M}& a_{1,M+1}& a_{1,M+2}& \cdots& a_{1,M+N}\\
	a_{2,1}&\cdots& a_{2,M}& a_{2,M+1}& a_{2,M+2}& \cdots& a_{2,M+N}\\
	\vdots&\ddots&\vdots&\vdots&\vdots&\ddots&\vdots\\
	a_{M,1}&\cdots& a_{M,M}& a_{M,M+1}& a_{M,M+2}& \cdots& a_{M,M+N}\\
	\end{matrix}
	\right],
	\end{equation}
	     where $a_{i,i}=0$,$i=1,2,...M$. $a_{i,j}$ denotes antenna number of the communication task between $i$-th platform and the $j$-th platform when $i\not=j, j=1,2,...,M$. It also denotes antenna number for the $(j-M)$-th localization task of the $i$-th platform when $i\not=j, j=M+1,M+2,...,M+N$.  
	
	\subsection{Communication model for array allocation}
	Communication capacity among different platforms is selected as one of the objective functions here.  
	The communication capacity between the $i$-th and $j$-th platform, denoted by $C(a_{i,j})$, can be expressed as follows \cite{Su2013,5469927}
	\begin{equation}
	C(a_{i,j})=\log_2\left[\det\left(\bI_{a_{i,j}}+\frac{P_{i,j}\bH\bH^{H}}{a_{i,j}\mathcal{N}_0}\right)\right],\\  {i,j=1,2,...,M,i\neq j,} \label{C1}
	\end{equation}
	where $\bI_{a_{i,j}}$ is the identity matrix of size $a_{i,j}$, $P_{i,j}$ denotes the transmit power using $a_{i,j}$ antennas, $\mathcal{N}_0$ is the channel noise covariance, and  $\bH$ is the channel coefficient matrix. The element $h_{m,n}$ of $\bH$ is given by \cite{Su2013}
	\begin{equation}
	h_{m,n}=\frac{\lambda}{r_{m,m}}\exp(-j2\pi\frac{r_{m,n}}{\lambda}), \label{h1}
	\end{equation}
	where $\lambda$ is the waveform length, $r_{m,m}$ is the distance from the $m$-th transmit antenna on the $i$-th platform to the $n$-th receiver antenna on the $j$-th platform.
	  Although the sparse array provides much more degrees of freedom for antenna configuration, it is also associated with unpredictable sidelobe behavior, which brings high complexity in beampattern design.    
	Here, the transmit and receiver antennas on a platform are equally spaced, while the distance between the $i$-th platform and the $j$-th platform is $d_{i,j}$ meters, where $d_{i,j}\gg d$. Therefore, all elements of $\bH$ are assumed to be the equal and equation \eqref{h1} can be written as
	\begin{equation}
	h_{m,n}=\frac{\lambda}{d_{i,j}}\exp(-j2\pi\frac{d_{i,j}}{\lambda}). 
	\end{equation}
	Then, equation \eqref{C1} can be formulated as:
	\begin{equation*}
	C(a_{i,j})=\log_2\left[\det\left(\bI_{a_{i,j}}+\left(\frac{\lambda}{d_{i,j}}\right)^2\textbf{1}_{a_{i,j}\times a_{i,j}}\frac{a_{i,j}\Delta P}{\mathcal{N}_0}\right)\right],\\
	{i,j=1,2...M, i\neq j.}\ineqno\label{comu1}
	\end{equation*}
	 where $\Delta P$ represents the transmit power of every antenna. 
	\subsection{Radars' localization model for array allocation}
	     In the platform network of this paper, the orthogonal waveform for target localization is assumed to be transmitted by platform $i$, reflected by the target, and received by platform $i$ itself. The noise and clutters from other targets and platforms are assumed to be well suppressed by other filtering techniques. In the localization process, the variation of the targets’ center of mass,  is also supposed to be small with respect to the system resolution capabilities here \cite{Godrich2011}. 
	
	The Cramer-Rao bound (CRB) matrix, $\bE_{z}$, provides a lower bound for the localization mean-square error (MSE) of $z$-th target. The lower bound on the sum of localization MSE, $\bL_{z}$, is denoted as the trace of CRB matrix. The $\bE_{z}$ and $\bL_{z}$ can be respectively written as \cite{Godrich2011,poor2013introduction}
	\begin{equation}
	\bE_{z}=\left[\sum_{i=1}^{M}\bp_{z}(a_{i,M+z})\bG\right]^{-1}, i=1,2...M,\;\; z=1,2...N
	\end{equation} 
	\begin{dmath}
		\bL_z(a_{1,M+z},...,a_{M,M+z})=\Tr(\bE_{z})
		=\frac{\bb^T\bp_{z}(a_{1,M+z},...,a_{M,M+z})}{\bp_{z}(a_{1,M+z},...,a_{M,M+z})^T\bQ\bp_{z}(a_{1,M+z},...,a_{M,M+z})},\label{rad1}
	\end{dmath}
	where $\bG$ and $\bQ$ are the system parameter matrices defined in \cite{Godrich2011,poor2013introduction}. Here, $\bp_{z}(a_{i,M+z})$ is an element of $\bp_{z}(a_{1,M+z},...,a_{M,M+z})$, which is the transmit power of the sensor network for localizing the $z$-th target, and can be written as
	\begin{equation}
	\bp_{z}(a_{1,M+z},...,a_{M,M+z})=\Delta P\left[a_{1,M+z},...,a_{M,M+z}\right]^T. \label{rad2}
	\end{equation}
	 
	\section{Optimal Array allocation for the RCI network }\label{sec3}
	     Based on the array resource allocation models for localization and communication, the proposed mathematical problems will be further derived based on Minorization Maximization (MM) method. The projected gradient descent (PGD) method is then used to solve the array allocation problems for the RCI network, in order to improve the localization and communication performance.  
	
	\subsection{MM Theroy}
	MM is the dual method for majorization minimization \cite{sun2017majorization}, a powerful minimization technique widely applied to engineering applications. Suppose we seek the solution of following optimization
	\begin{gather}\label{ha6}
	\max_{\bx}f(\bx)\\
	s.t. \;\bx \in \mathcal{X},
	\end{gather}
	where $ \bx $ is decision variable, $ \mathcal{X} \subset \mathbb{R}^{n} $ is feasible set and $ f(.) $ is the objective. The MM
	propose optimizing a minorizer problem iteratively:
	$ \bx^{m+1} = \arg\max_{\bx}g(\bx|\bx^{m}) $, where the minorizer function $ g(.|\bx^{m}) $ satisfies the following golden conditions:
	\begin{gather}\label{key}
	1) f(\bx)\ge g(\bx|\bx^{m}), \bx,\bx^{m}\in\mathcal{X},\\
	2) f(\bx^{m}) = g(\bx^{m}|\bx^{m}), \bx^{m}\in\mathcal{X}.
	\end{gather}
	It is easy to verify that the sequence $ \{\bx^{m}\} $ converges to the optimal solution of \eqref{ha6} by increasing the objective in each step \cite{sun2017majorization}.      
	In fact,  $ f(\bx^{m+1}) \ge g(\bx^{m+1}|\bx^{m})\ge g(\bx^{m}|\bx^{m}) = f(\bx^{m}) $ which ensures a non-decreasing path for the sequence.
	The following Lemma is applied for finding the majorizer in this paper. 
	 
	\begin{lemma}\label{lem}
		Suppose $ f:\mathbb{R}^{n}\rightarrow\mathbb{R} $ is a continuously differentiable function with a Lipschitz continuous gradient. Also, suppose that there exist a matrix $ \bM  $ such that  $  \bM\le \nabla^{2}f(\bx) $. That is, $  \nabla^{2}f(\bx)-\bM  $ positive semi-definite for every $ \bx \in \mathbb{R}^{n} $. Then, we have
		\begin{equation}\label{ha1}
		{f}(\bx) \ge {f}(\by)+\nabla{f}(\by)^{T}(\bx-\by)+\frac{1}{2}(\bx-\by)^{T}{\bM}(\bx-\by),
		\end{equation}
		for every $ \bx, \by \in \mathbb{R}^{n}$
	\end{lemma}   
	\begin{proof}
		This Lemma is the dual form of Lemma 12 in \cite{sun2017majorization}. To observe the proof, define $ \bar{f}(.):= -f(.) $ and $ \bar{\bM} = -\bM $, then $ \bar{f}(.)$ and $ \bar{\bM} $ satisfy the assumptions in Lemma 12 of \cite{sun2017majorization}. Therefore, we have
		\begin{equation}\label{key}
		\bar{f}(\bx) \le \bar{f}(\by)+\nabla\bar{f}(\by)^{T}(\bx-\by)+\frac{1}{2}(\bx-\by)^{T}\bar{\bM}(\bx-\by),
		\end{equation}
		for every $ \bx, \by \in \mathbb{R}^{n}$. By multiplying both side to $ -1 $, \eqref{ha1} is derived.
	\end{proof}
	 
	\subsection{Maximizing  channel capacities and minimizing CRLB}
	     Two basic problems are given as follows  with the constraint of total number of antennas to design the optimal array allocation strategy:
	\begin{gather}\label{ha12}
	\max_{\bA} \;\;C(a_{i,j})\;\; i,j = 1,...M, i\ne j, \nonumber \\
	\max_{\bA} \;\;\frac{1}{\bL_z(a_{1,M+z},...,a_{M,M+z})} \;\; z=1,2,...,N,   \label{c1}\\
	s.t. \sum_{j=1}^{M+N} a_{i,j}\leq \gamma_i,\;\;
	a_{i,j}\in \mathbb{N}, \;\; \bA\in \mathbb{N}^{M\times(M+N)}, \nonumber
	\end{gather} 
	where $C(a_{i,j})$ is the communication capacity from the $i$-th platform to the $j$-th platform, $\bL_z(a_{1,M+z},...,a_{M,M+z})$ is defined in \eqref{rad1}  for the $z$-th target, $\gamma_i$ is the total transmit antennas number of the $i$-th platform. The objective functions in our algorithm, will be designed based the scalarization of performance metrics for communication and localization in the network, which have been defined here.
	 
	In \eqref{c1}, the constraint function can be simplified as:
	\begin{equation}
	\sum_{i=1}^{M}\sum_{j=1}^{M+N}a_{i,j}=\textbf{1}_{M}^T\bA\textbf{1}_{M+N}, 
	\end{equation}
	First note that $ \log_2(.) $ is an increasing function. Therefore, this function can be omitted in maximizing each $ C(a_{i,j}) $. 
	Consequently, by combining \eqref{comu1}, \eqref{rad1} and \eqref{rad2}, the  multi-objective problem can be rewritten as  
	\begin{gather}
	\max_{\bA}\;\;  \det\left(\bI_{a_{i,j}}+\left(\frac{\lambda}{d_{i,j}}\right)^2\textbf{1}_{a_{i,j}\times a_{i,j}}\frac{a_{i,j}\Delta P}{\mathcal{N}_0}\right),\label{o2_1}  \\ i\ne j = 1,...M, \nonumber \\ 
	\max_{\bA} \;\;\frac{\Delta P\left[a_{1,M+z},...,a_{M,M+z}\right]\bQ\left[a_{1,M+z},...,a_{M,M+z}\right]^T}{\bb^T\left[a_{1,M+z},...,a_{M,M+z}\right]^T}\nonumber\\ z=1,2,...,N,      \label{o2_2}  \\
	s.t. \;\;\textbf{1}_{M}^T\bA\textbf{1}_{M+N}\leq \gamma_i \;\;
	a_{i,j}\in \mathbb{N}. \label{oc1}
	\end{gather}
	To further simplify the problem, we employ the following approximation from  \cite{petersen2012matrix} for small values of $\varepsilon$:
	\begin{equation}
	\det(\bI+\varepsilon\bX)\cong 1+\det(\bX)+\varepsilon\Tr(\bX)+\frac{1}{2}\varepsilon^2\Tr(\bX)^2-\frac{1}{2}\varepsilon^2\Tr(\bX^2).
	\end{equation}
	In this regard, note that 
	$ \frac{\lambda}{d_{i,j}}\ll1 $.
	Therefore, we have
	\begin{multline}
	\det\left(\bI_{a_{i,j}}+\left(\frac{\lambda}{d_{i,j}}\right)^2\textbf{1}_{a_{i,j}\times a_{i,j}}\frac{a_{i,j}\Delta P}{\mathcal{N}_0}\right) \\
	\cong 1+\det\left(\textbf{1}_{a_{i,j}\times a_{i,j}}\frac{a_{i,j}\Delta P}{\mathcal{N}_0}\right)
	+\left(\frac{\lambda}{d_{i,j}}\right)^2\Tr\left(\textbf{1}_{a_{i,j}\times a_{i,j}}\frac{a_{i,j}\Delta P}{\mathcal{N}_0}\right)\\
	+\frac{1}{2}\left(\frac{\lambda}{d_{i,j}}\right)^4\Tr\left(\textbf{1}_{a_{i,j}\times a_{i,j}}\frac{a_{i,j}\Delta P}{\mathcal{N}_0}\right)^2
	-\frac{1}{2}\left(\frac{\lambda}{d_{i,j}}\right)^4\Tr\left(\textbf{1}_{a_{i,j}\times a_{i,j}}\textbf{1}_{a_{i,j}\times a_{i,j}}\left(\frac{a_{i,j}\Delta P}{\mathcal{N}_0}\right)^2\right)\\
	\simeq 1+\left(\frac{\lambda}{d_{i,j}}\right)^2\frac{(a_{i,j})^2\Delta P}{\mathcal{N}_0} \label{o12}.
	\end{multline}
	%
	Hence, the first set of objectives, i.e., \eqref{o2_2} can be rephrased as the following multi-objective problem:
	\begin{equation}\label{key}
	\max_{\bA}\left[\left(\frac{\lambda}{d_{i,j}}\right)^2\frac{(a_{i,j})^2\Delta P}{\mathcal{N}_0};i,j = 1,...M, i\ne j\right].
	\end{equation}
	This problem can be converted to a single-objective problem by using scalarization. Define  $\bx=vec(\bA)$; then, this problem is equivalent to
	\begin{gather}
	\max_{\bA} \bmu^{T}\left(\frac{\Delta P}{\mathcal{N}_0}vec(\bLambda)\circ\bx^{\circ2} \right) \nonumber
	\\=\bmu^{T}\left[\frac{\Delta P}{\mathcal{N}_0}\Diag(\Diag(vec(\bLambda))\bx)\bx\right],\label{j1}
	\end{gather}
	where $ \bmu = [\mu_{1}, \mu_{2}, ..., \mu_{M\times(M+N)}] $ is the scalarization coefficient vector, determined by the relative importance of objectives. $ \ba \circ \bb $ and $ \bx^{\circ 2} $ represent Hadamard product and power, respectively. Also, $\Diag(.)$ is the main diagonal linear operator, and
	\begin{gather}
	\bLambda=
	\left[
	\begin{matrix}
	\left(\frac{\lambda}{d_{1,1}}\right)^2&    \left(\frac{\lambda}{d_{1,2}}\right)^2&\cdots&     \left(\frac{\lambda}{d_{1,M}}\right)^2& \boldsymbol{0}_{N}^T\\
	\left(\frac{\lambda}{d_{2,1}}\right)^2&    \left(\frac{\lambda}{d_{2,2}}\right)^2&\cdots&     \left(\frac{\lambda}{d_{2,M}}\right)^2& \boldsymbol{0}_{N}^T\\\\
	\vdots&\vdots&\ddots&\vdots& \boldsymbol{0}_{N}^T\\
	\left(\frac{\lambda}{d_{M,1}}\right)^2&    \left(\frac{\lambda}{d_{M,2}}\right)^2&\cdots&     \left(\frac{\lambda}{d_{M,M}}\right)^2& \boldsymbol{0}_{N}^T\\
	\end{matrix}
	\right]. \label{canshu1}
	\end{gather}

	Furthermore, the second set of objectives, i.e., \eqref{o2_2} can be also rephrased as:
	\begin{equation}
	\max_{\bA} \left[\frac{\be_z^T\bA^T\bQ\bA\be_z}{\bb^T\bA\be_z}; z=1,2,...,N \right],   \label{o11}
	\end{equation}
	where 
	\begin{equation}\label{key}
	\be_z=[0,0,...,1_{(M+z)},...,0]^T.
	\end{equation}
	In the  denominator of \eqref{o11}, we have
	\begin{equation}
	\bb^T\bA\be_z=\Tr(\bb^T\bA\be_z)
	=\Tr(\be_z\bb^T\bA)
	=vec(\bb\be_z^T)^Tvec(\bA)
	=\bk^T\bx, \label{j2}
	\end{equation}
	where $\bk=vec(\bb\be_z^T)$ and $\bx=vec(\bA)$.
	The numerator of \eqref{o11} can also be simplified:
	\begin{multline}
	\be_z^T\bA^T\bQ\bA\be_z
	=\Tr(\be_z^T\bA^T\bQ\bA\be_z)
	=\Tr(\bA^T\bQ\bA\be_z\be_z^T)
	=vec(\bA)^Tvec(\bQ\bA\be_z\be_z^T)\\
	=vec(\bA)^T(\be_z\be_z^T\otimes\bQ)vec(\bA)
	=\bx^T(\be_z\be_z^T\otimes\bQ)\bx
	=\bx^T\bB_z\bx,  \label{j3}
	\end{multline}
	where
	\begin{equation}\label{key}
	\bB_z=\be_z\be_z^T\otimes\bQ.
	\end{equation}
	Since the decision variable, $ \bx = vec(\bA) $, appears both in numerator and denominator, we apply the $ \ln(.) $ function on the objectives for $ z = 1,...,N $. Note that the resulting problem is equivalent to \eqref{o11} because $ \ln(.) $ is an increasing function. After scalarization the problem is given by
	\begin{equation}\label{key}
	\max_{\bx} \sum_{z = 1}^{N}\zeta_{z}\ln\left(\frac{\bx^{T}\bB_z\bx}{\bk^{T}\bx}\right)
	\end{equation}
	where $ \zeta_{z}, z = 1,...,N$, are scalarization coefficients defining relative importance of the objectives.   The order to describe the optimization problems more clearly, the tasks in the resource scheduling problems are often assumed to be the same important in for designing the objective function \cite{Feng2016}\cite{8368273} or constraint condition \cite{Xie2018}. So all channel capacities and radar antenna links are assumed to be equally important in our algorithm, and scalarization coefficients are all set to one. Without loss of generality, we continue our development by assuming $ \bmu = \textbf{1}_{M \times (M+N)} $ and $ \zeta_{z} = 1, z = 1,...,N $; other scenarios can be developed by following the same steps.  
	
	Moreover, equation \eqref{oc1} can be expressed in matrix vector format:
	\begin{equation}
	\bA\textbf{1}_{M+N}\leq \bgamma_0, \label{oc3}
	\end{equation}
	where $\bgamma_0=[\gamma_1,\gamma_2,...,\gamma_M]^T$ and $ \le $ denotes element-wise relation.
	Based on the $ vec(.) $ operator in \cite{petersen2012matrix}, \eqref{oc3} can also be expressed as:
	\begin{equation}
	\bV_0\bx\leq\bgamma_0,
	\end{equation}
	where 
	\begin{gather}
	\bV_0=
	\left[
	\begin{matrix}
	1&\textbf{0}_{M-1}^T&\cdots& 1&\textbf{0}_{M-1}^T\\
	0& 1&\textbf{0}_{M-2}^T&\cdots& \boldsymbol{0}_{M-2}^T\\
	\ddots&\ddots&\ddots&\ddots&\ddots\\
	\textbf{0}_{M-1}^T&    1&\cdots&     \textbf{0}_{M-1}^T&    1\\
	\end{matrix}
	\right]_{M\times [M(M+N)]},
	\end{gather}\\
	and in the first row of $\bV_0$, there are $M+N$ groups of $[1,\textbf{0}_{M-1}^T]$.
	
	     To simplify notations, we define the following functions as the objective functions:  
	\begin{gather}
	f_1=\frac{\Delta P}{\mathcal{N}_0}\textbf{1}_{M\times(M+N)}^T\left[\Diag(\Diag(vec(\bLambda))\bx)\bx\right],\\
	f_2=\sum_{z=1}^{N}\left(\ln(\bx^T\bB_z\bx)-\ln(\bk^T\bx)\right)
	=\sum_{z=1}^{N} f_{2z},\label{ha5}
	\end{gather}
	where $f_{2z}:=\left(\ln(\bx^T\bB_z\bx)-\ln(\bk^T\bx)\right)$.
	According to the MM theory \cite{qiu2016prime,sun2017majorization,hunter2004tutorial}, we construct the following quadratic functions $g_1(.|\bx^{m})$ and $g_2(.|\bx^{m})$ as the minorizer functions:
	\begin{gather}
	g_1(\bx|\bx^{m})=f_1(\bx^m)+\nabla f_1(\bx^m)^T(\bx-\bx^m) \nonumber \\
	+\frac{1}{2}(\bx-\bx^m)^T\bM_1(\bx-\bx^m)  \label{mm1}
	\end{gather}
	and
	\begin{gather}
	g_2(\bx|\bx^{m})=\sum_{z=1}^{N} [f_{2z}(\bx^m)+\nabla f_{2z}(\bx^m)^T(\bx-\bx^m)  \nonumber \\
	+\frac{1}{2}(\bx-\bx^m)^T\bM_{2z}(\bx-\bx^m)], \label{mm2}
	\end{gather}
	where
	\begin{gather}
	\nabla f_1(\bx)=2\frac{\Delta P}{N_0}\Diag(vec(\bLambda))\bx^m,\\
	\bM_1=\frac{\Delta P}{N_0}\Diag(vec(\bLambda)),\\
	\nabla f_{2z}=\frac{(\bB_z+\bB_z^T)\bx}{\bx^T\bB_z\bx}-\frac{\bk}{\bk^T\bx},
	\end{gather}
	and 
	\begin{equation*}\label{key}
	\bM_{2z}=\frac{\bB_z+\bB_z^T}{(\bx^m)^T\bB_z\bx^m}-\frac{(\bB_z+\bB_z^T)\bx^m(\bx^m)^T(\bB_z+\bB_z^T)}{((\bx^m)^T\bB_z\bx^m)^2}-2\frac{\bk\bk^T}{(\bk^T\bx^m)^2}.    
	\end{equation*}
	Consequently, we have $  \bM_1\le \nabla^2 f_1(\bx^m)$ and $ \bM_{2z}\le \nabla^2 f_{2z}(\bx^m) $. Using Lemma \ref{lem}, we find that  $ f_{1}(.) $ and $ f_{2}(.) $ are minorized by $g_1(.|\bx^{m})$ and $g_2(.|\bx^{m})$, respectively.
	Omitting the constants  and combining the two objective functions using  scalarization, the final objective function can be written as:
	\begin{gather}
	\max_{\bx} \Psi_1(\bx) \nonumber\\
	s.t.\;\; \bx\in \mathbb{N},\;\; \bV_0\bx\leq\bgamma_0, \label{o1c1}
	\end{gather}
	where
	\begin{multline}
	\Psi_1(\bx)=\nabla f_1(\bx^m)(\bx-\bx^m)+\frac{1}{2}(\bx-\bx^m)^T\bM_1(\bx-\bx^m)\\
	+w_0\left[\sum_{z=1}^{N}\left(\nabla f_{2z}(\bx^m)^T(\bx-\bx^m)+\frac{1}{2}(\bx-\bx^m)^T\bM_{2z}(\bx-\bx^m)\right)\right]   \label{oo1}
	\end{multline}
	and $w_0$ is the unification weight or scalarization coefficient. The
	PGD method \cite{7373645,7932172} is employed here to conclude the solution of this array allocation problem.
	The gradient of $\Psi_1(\bx)$ can be written as:
	\begin{multline}
	\nabla \Psi_1(\bx)=\nabla f_1(\bx^m)+\frac{1}{2}(\bM_1+\bM_1^T)(\bx-\bx^m)\\
	+w_0\left[\sum_{z=1}^{N}\left(\nabla f_{2z}(\bx^m)^T+\frac{1}{2}(\bM_{2z}+\bM_{2z}^T)(\bx-\bx^m)\right)\right].
	\end{multline}
	Using the PGD, $\bx_k$ can be computed iteratively by:
	\begin{equation}
	\bx_k=pro_{\bx}(\bx_{k-1}-t_k\nabla\Psi_1(\bx_{k-1})),\;\; k=1,2,3...   \label{p1}
	\end{equation}
	where $pro_{\bx}(\bx)$ is the Euclidean projection and is derived by:
	\begin{gather}
	pro_{\bx}(\bx)=\arg\min_{\bz}\|\bz-\bx\|_2^2   \label{p2}\\
	s.t. \bx\in \mathbb{N},\;\; \bV_0\bx\leq\bgamma_0. \nonumber
	\end{gather}
	In summary, the first array resource allocation procedures in the RCI network is described in Algorithm \ref{alg1}.

	\begin{algorithm}[!h]
		\KwIn{Information of localization and communication targets, Number of iterations ($K_1$,$K_2$)    }
		\KwResult{Allocated antennas}
		
		1. Set $ m = 0 $; start from one random sequence  $ \bx^{(0)} $ which satisfies the constraint in (\ref{o1c1}). \\

		2. \Repeat{$ m \ge K_1 $}{
			
			3. Construct surrogate functions, $ g_1$ and $ g_2 $,  using formulas \eqref{mm1} and \eqref{mm2}
			
			4.  Design the objective function in \eqref{o1c1}
			
			5. Set $ k = 0 $ and $ \bx_{(0)} $ as the initial vectors
			
			6. \Repeat{$ k\ge K_2 $}{
				
				7.    Compute the intermediate variable $ \br_k = \bx_{k-1}-t_k\nabla\Psi_1(\bx_{k-1})$

				8. Obtain $ \bx_{k} $ after projection based \eqref{p2}
				
				9. Round each element of $ \bx_{k} $ to the nearest integer number
				
				10. $ k = k+1 $
				
			}
			
			11. $ m = m+1 $
		}
		\caption{Array Allocation  for Radar and Communication Integration (A2RCI) }\label{alg1}
	\end{algorithm}
	
	\begin{remark}
		The optimization in \eqref{c1}  is a NP-hard problem. To observe this, note that they are instances of integer programming because of their criterion; it is NP-complete and  NP-hard. Interestingly, MM has been applied to solve NP-hard optimization problems previously \cite{Esmaeili-Najafabadi2019, Cui2017, wu2017cognitive} with very satisfying results, and it is naturally expected to have promising results here as well. To further analyze the algorithm, note that it consists of two loops: an outer loop derived from MM and an inner loop derived from PGD. The MM guarantees a non-decreasing path and hence the convergence. However, this convergence might be to a local maximum. The PGD part can avoid this because it is using a limited number of steps, i.e., for $ k\le K_2 $, if the $ t_k $ is properly tuned.
	\end{remark}
	
	\subsection{Allocating Communication Channels First}
	In this section, another algorithm is proposed  to optimize the radar localization performance that guarantees the desired communication performance. Henceforth, we express the design problem in terms of $ \bA $ or $ \bx = vec(\bA) $ whenever necessary  to avoid repeating the derivations. Consider the following design problem
	\begin{gather}
	\max_{\bx}\;       f_2(\bx)        \label{o2}\\       
	s.t.\; \log_2\left[\det\left(\bI_{a_{i,j}}+\left(\frac{\lambda}{d_{i,j}}\right)^2[\textbf{1}]_{a_{i,j}\times a_{i,j}}\frac{a_{i,j}\Delta P}{\mathcal{N}_0}\right)\right]\le \eta,  \label{o21} \\
	i,j=1,2,...,M,i\neq j,       \nonumber        \\
	\bx\in \mathbb{N},\;\; \bV_0\bx\leq\bgamma_{0}, 
	\end{gather}
	where $\eta$ is the desired threshold for every communication channel and $ f_2(.) $ is defined in \eqref{ha5}. As mentioned above, $ f_{2}(.) $ is minorized by 
	\begin{equation}
	\Psi_2(\bx)= g_2(\bx|\bx^{m})  \nonumber\\
	=\sum_{z=1}^{N}\left( f_{2z}(\bx^m)+\nabla f_{2z}(\bx^m)^T(\bx-\bx^m)+\frac{1}{2}(\bx-\bx^m)^T\bM_{2z}(\bx-\bx^m) \right).\label{ha4}
	\end{equation}
	We divide the minorizer problem into two subproblems and cyclically update the result. 
	
	First, the capacity constraint is separated from the problem. Following the same steps introduced in \eqref{o12} and \eqref{j1}, equation \eqref{o21} can be converted to
	\begin{equation}
	\frac{\Delta P}{\mathcal{N}_0}\bLambda\circ\bA^{\circ2}\le \bF,    \label{o2c2}
	\end{equation}
	where
	\begin{gather}
	\bF=[2^{\eta-1}\textbf{1}_{M\times M},\boldsymbol{0}_{M,N}].
	\end{gather}
	Also, \eqref{o2c2} can be rewritten as:
	\begin{equation}
	vec(\bF-\frac{\Delta P}{\mathcal{N}_0}\bLambda\circ\bA^{\circ2})
	=
	\textbf{f}-\frac{\Delta P}{\mathcal{N}_0}vec(\bLambda)\circ\bx^{\circ2}\le \bzero_{M^{2}+MN} ,   \label{canshu2}
	\end{equation}
	where 
	\begin{equation}\label{h2}
	\textbf{f}=vec(\bF).
	\end{equation}
	We further suppose that each channel first gets its maximum allowable resource. In this case, the inequality in \eqref{canshu2} converts to equality and we have
	\begin{equation}\label{ha3}
	\bx^{*} = \left(\frac{\mathcal{N}_0}{\Delta P}\ff \oslash vec(\bLambda)\right)^{\circ\frac{1}{2}},
	\end{equation}
	where $ \oslash $ denotes element-wise division. Note that this amount represents array resource dedicated to communication. The remaining resource can define the available resources for radar. 
	
	After removing the capacity constraint, the minorizer problem is of the form
	\begin{gather}
	\max_{\bx}\;      \Psi_2(\bx)    \label{o2cc},  \\
	s.t. \;\;\bx\in \mathbb{N},\;\; \bV_1\bx\leq\bgamma_1, \nonumber 
	\end{gather}
	where $ \bgamma_1 $ is determined by the remaining resource for the radar. That is, it is obtained by $ \bgamma_{1}=\bgamma_{0} - \bV_{0}\bx^{*} $. Also, $ \bV_{1} $ is given by
	\begin{gather}
	\bV_1=
	\left[
	\begin{matrix}
	\textbf{0}_{M^2}^T&1&\textbf{0}_{M-1}^T&\cdots& 1&\textbf{0}_{M-1}^T\\
	\textbf{0}_{M^2}^T&0&1&\textbf{0}_{M-2}^T&\cdots& \boldsymbol{0}_{M-2}^T\\
	\vdots&\ddots&\ddots&\ddots&\ddots&\ddots\\
	\textbf{0}_{M^2}^T&\textbf{0}_{M-1}^T&1&\cdots&     \textbf{0}_{M-1}^T&    1\\
	\end{matrix}
	\right].
	\end{gather}
	where its size is ${M\times[M(M+N)]}$, and in the first row of $\bV_1$, there are $N$ groups of $[1,\textbf{0}_{M-1}^T]$. Finally, the problem in \eqref{o2cc}  can be solved by using PGD as in \eqref{p1} and \eqref{p2}.
	Consequently, the second array resource allocation strategy for the RCI network is presented in Algorithm \ref{alg2}.
	
	\begin{algorithm}[!h]
		\KwIn{Information of localization and communication targets, Number of iterations ($K_1$,$K_2$), desird threshold $\eta$ for communication    }
		\KwResult{Allocated antennas}
		1. Construct $\bLambda$ and $\textbf{f}$ using \eqref{canshu1} and \eqref{h2}
		
		2. Allocate array resource for communication using \eqref{ha3}
		
		3. Update the left array resource $\bgamma_1$
		
		4. Set $ m = 0 $; start from one random sequence  $ \bx^{(0)} $ which satisfies the constraint $ \bV_1\bx\leq\bgamma_1 $ \\

		5. \Repeat{$ m \ge K_1 $}{
			
			6. Construct surrogate functions of $ g_2 $ using the formula of \eqref{mm2}
			
			7.  Design the objective function in \eqref{o2cc}
			
			8. Set $ k = 0 $ and $ \bx_{(0)} $ as the initial vectors
			
			9. \Repeat{$ k\ge K_2 $}{
				
				10.    Compute the intermediate variable $ \br_k = \bx_{k-1}-t_k\nabla\Psi_2(\bx_{k-1})$

				11. Obtain $ \bx_{k} $ after projection
				
				12. Round each element of $ \bx_{k} $ to the nearest integer number
				
				13. $ k = k+1 $
				
			}
			
			14. $ m = m+1 $
		}
		\caption{Array Allocation  for RCI , Algorithm 2 (A2RCI-II)}\label{alg2}
	\end{algorithm}
	
	\section{Simulations}\label{sec4}
	In this section, we present computer simulations to validate the efficiency of the proposed array allocation strategies. 
	In an RCI network, each platform is considered having the same uniform linear array comprising 600 transmit antennas. Without loss of generality, the transmit power of every antenna is set to $\Delta P= 1 $ kW. We only consider the allocation approaches for transmitting and suppose the receiver array resources are sufficient.
	All the simulations are analyzed and realized using the Matlab R2014a version, performing on a standard PC (with CPU Core i5, 2.4GHz, and 4GB of RAM). 
	
	To evaluate the proposed approach, we compare the localization and communication performance obtained by A2RCI (Algorithm \ref{alg1}), the nondominated sorting genetic algorithm II (NSGA-II) \cite{NSGA2019} and the multiple optimizations based on particle swarm optimization (MOPSO) \cite{MA2015}, which are the classical multi-objective optimization methods. 
	Three targets and three platforms are generated randomly over ten times in this simulation(Case-I). Figs. \ref{fig2} and \ref{fig3} show the comparison of communication capacity and localization 
	Root Cramer-Rao bound (RCRB).
	It is observed that Algorithm 1 has much better communication and localization performance due to the efficient array allocation.

	\begin{figure}\centering
		\includegraphics[width=.8\linewidth]{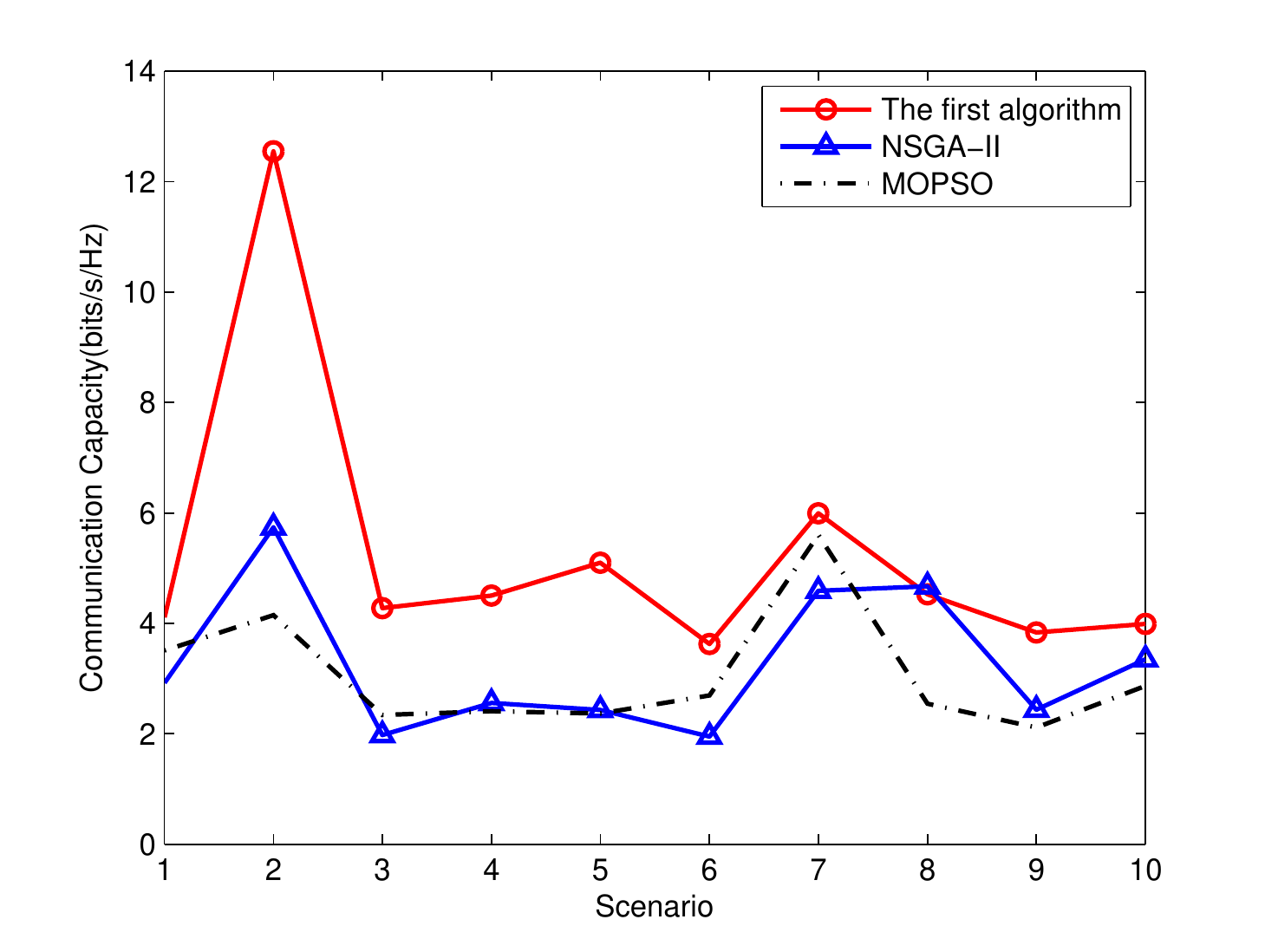}
		\caption{  Communication capacity comparison (3 targets and 3 platforms)}
		\label{fig2}
	\end{figure}
	
	\begin{figure}\centering
		\includegraphics[width=.8\linewidth]{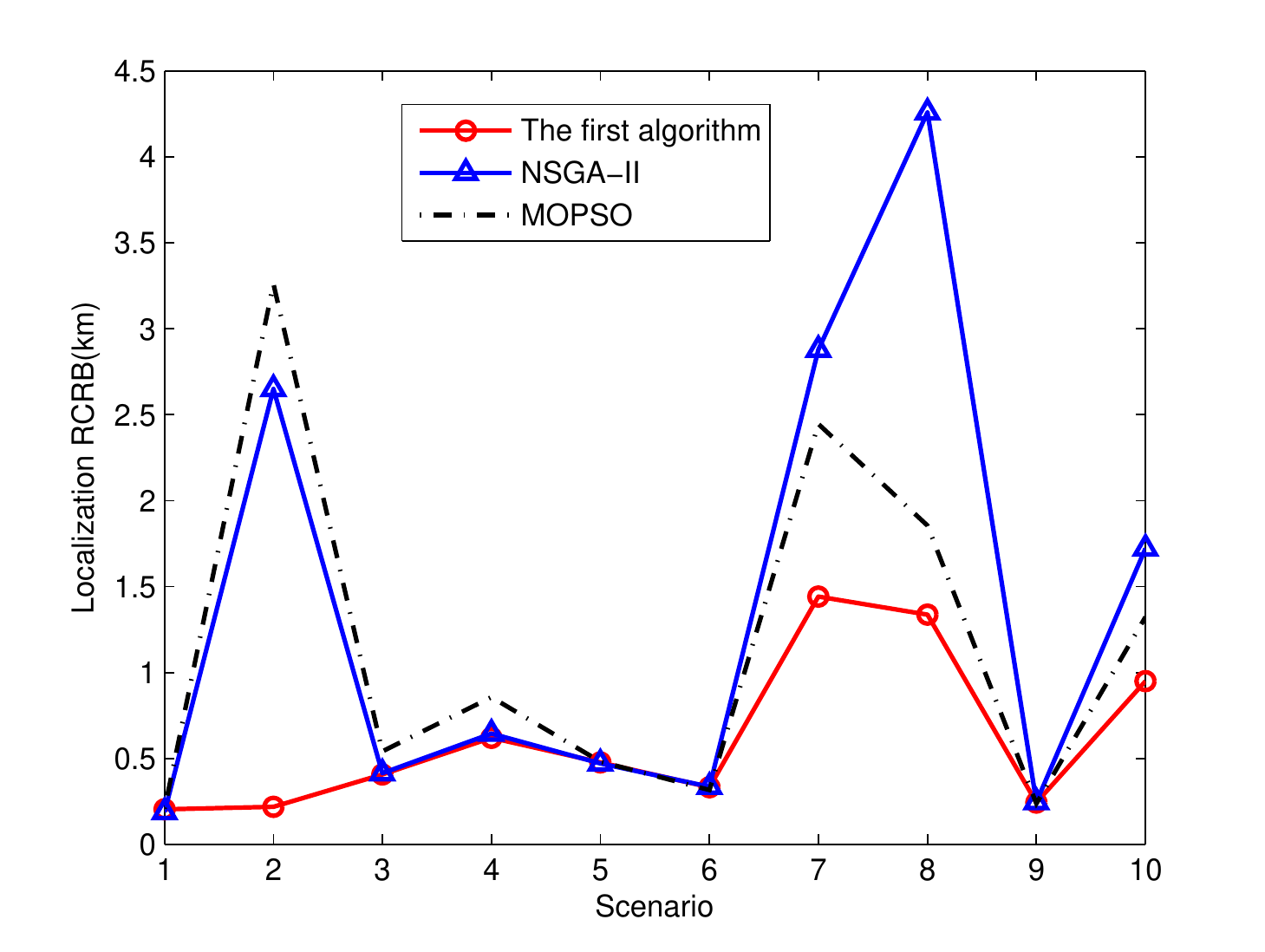}
		\captionof{figure}{  Localization performance comparison (3 targets and 3 platforms)}
		\label{fig3}
	\end{figure}
	
	Fig. \ref{fig4} depicts the radar and communication integration network with 12 platforms and 12 targets(Case-II). In this scenario, the communication capacity of every platform and the localization performance of each target are shown in Figs. \ref{fig5} and \ref{fig6}, respectively. The average communication capacity(ACC) and average localization CRLB (ALC) are also listed in Table 1. It can be observed that both of  NSGA-II and MOPSO based array allocation methods get worse performance than the proposed one. These results show that due to the increase of numbers of targets and platforms, there are many variables for the NSGA-II and MOPSO methods, which result in poorer performance. 
	
	\begin{center}
		\includegraphics[width=.8\linewidth]{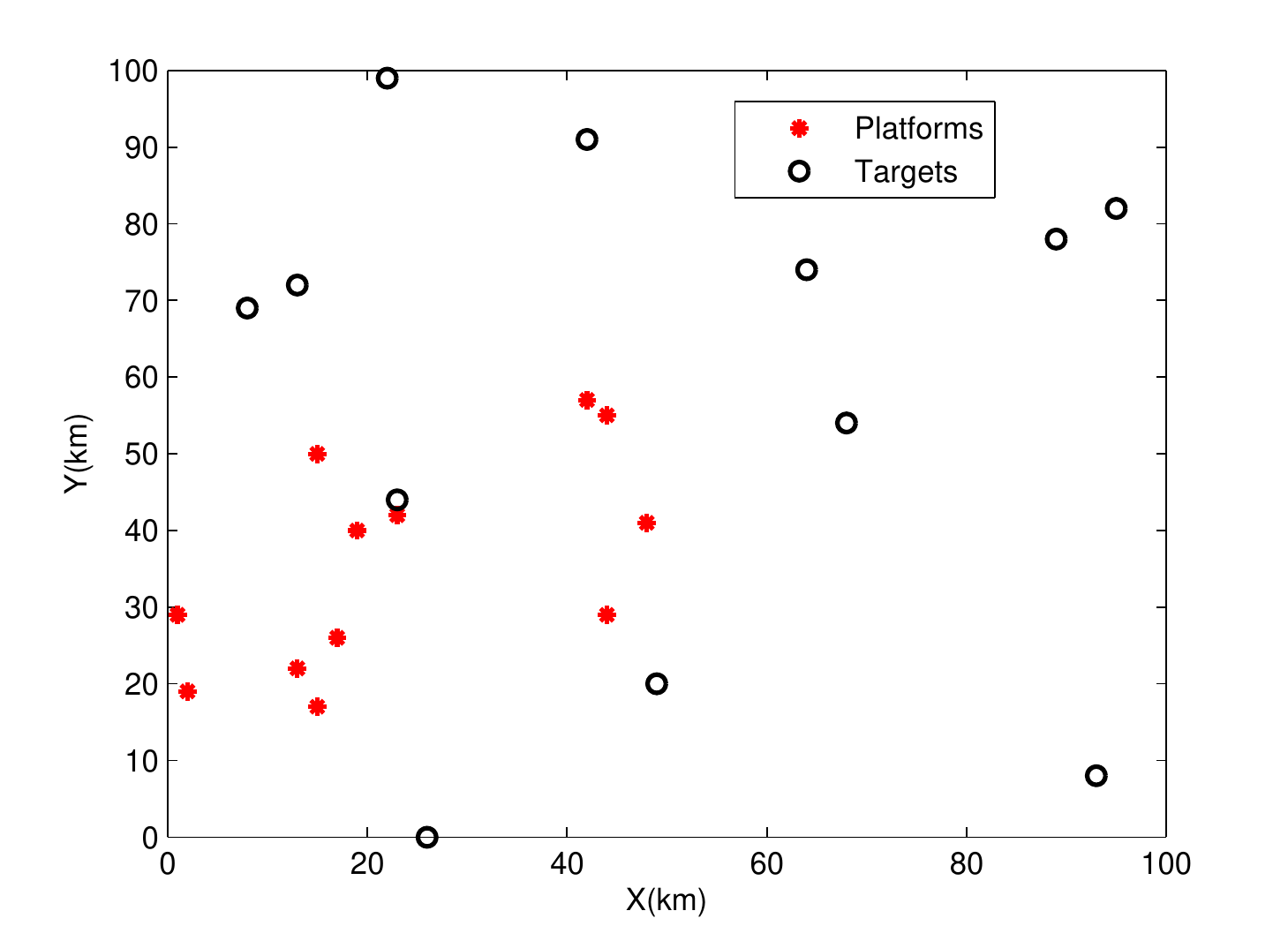}
		\captionof{figure}{  Radar and Communication integration network scenario (12 targets and 12 platforms)}
		\label{fig4}
	\end{center}
	
	\begin{center}
		\includegraphics[width=.8\linewidth]{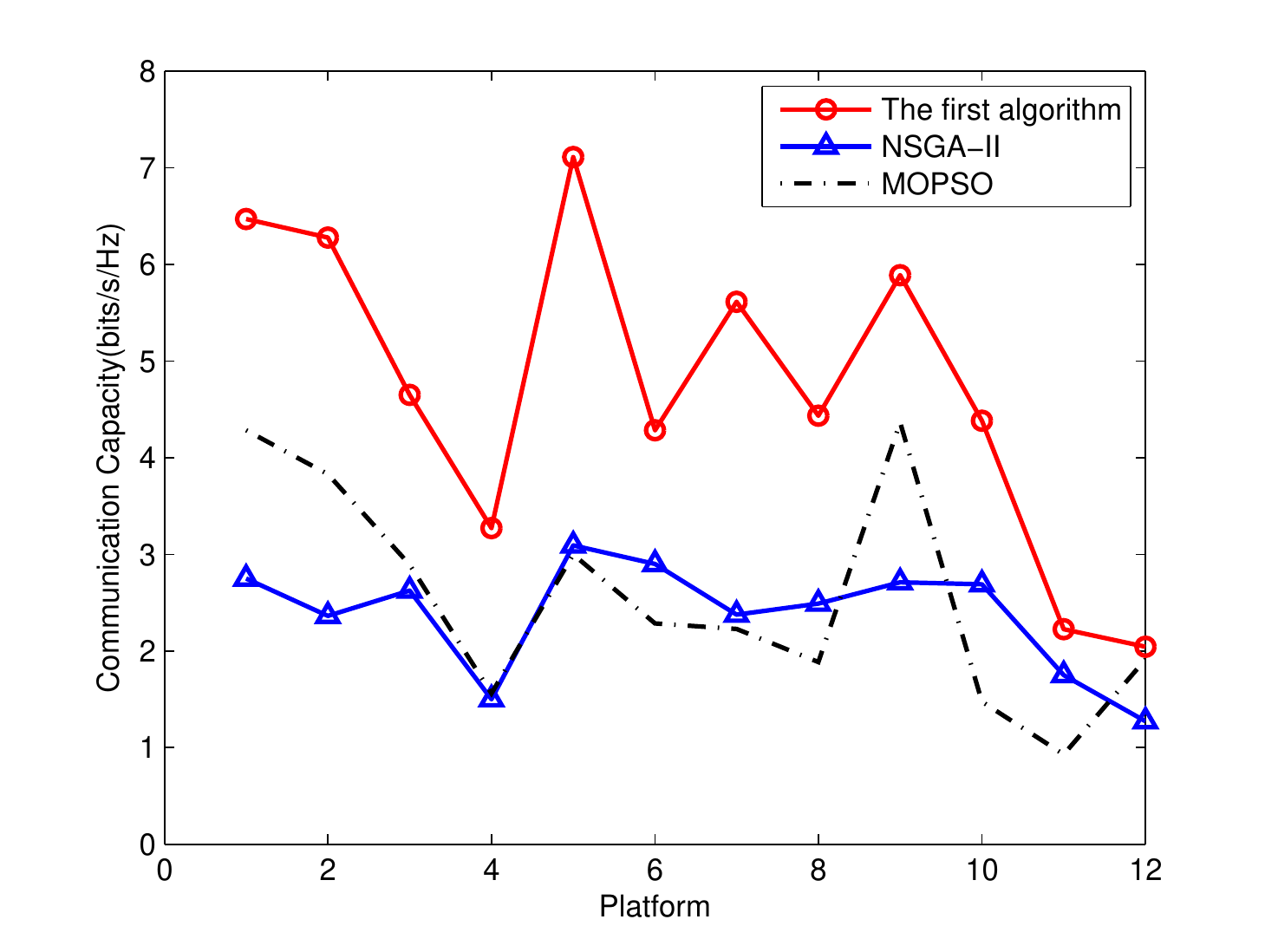}
		\captionof{figure}{  Communication Capacity comparison (12 targets and 12 platforms)}
		\label{fig5}
	\end{center}
	
	\begin{center}
		\includegraphics[width=.8\linewidth]{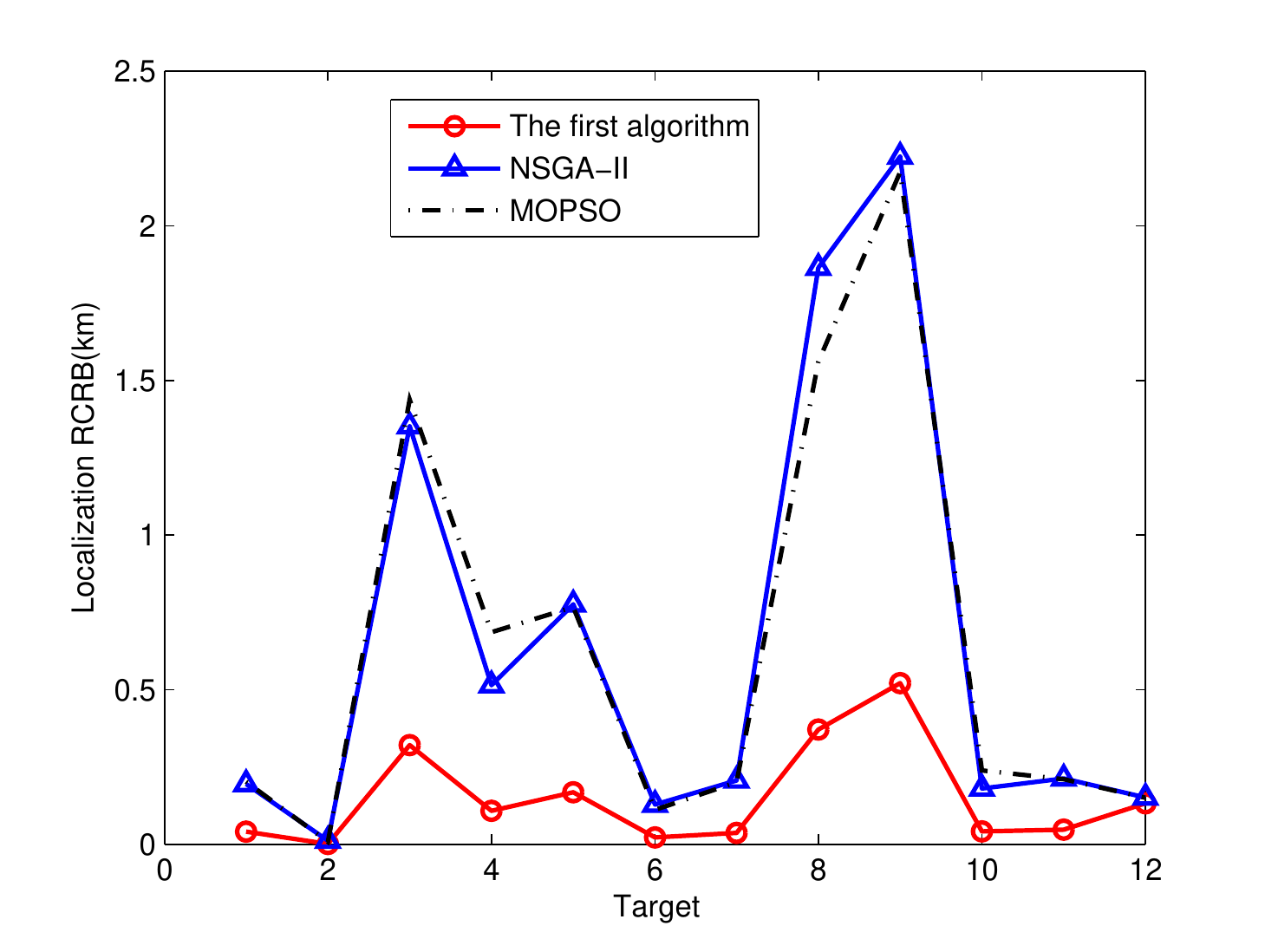}
		\captionof{figure}{  Localization performance comparison (12 targets and 12 platforms)}
		\label{fig6}
	\end{center}
	\begin{table*}[]    
		\centering    
		\caption{Performance comparisons based on Algorithm \ref{alg1}}    
		\label{Tab1}    
		\begin{tabular}{ccccccc}        
			\toprule        
			\multirow{2}{*}{Methods} & \multicolumn{2}{c}{$Case-I$} & \multicolumn{2}{c}{$Case-II$}\\        
			\cmidrule(r){2-3} \cmidrule(r){4-5}    
			&  $ACC(bits/s/Hz)$      &  $ALC(km)$   &   $ACC(bits/s/Hz)$     &   $ALC(km)$   \\        
			\midrule        
			$The$  $ first $ $ algorithm$             &5.251                         & 0.519                    & 4.721                   & 0.741\\        
			$NSGA-II$                         &3.263                       & 0.624                    & 2.277                 & 0.963  \\        
			$MOPSO$                           &3.061                          & 0.616                   & 2.554                 & 1.003\\        
			\bottomrule        
		\end{tabular}    
	\end{table*}
	
	Four different scenarios are considered in the second experiment, including three targets and three platforms, six targets and six platforms, nine targets and nine platforms, 12 targets, and 12 platforms. With 100 Monte Carlo simulations, Fig. \ref{fig7} shows the average communication capacity of each platform, while Fig. \ref{fig8} shows the average localization RCRB of each target. In Algorithm \ref{alg1}, the MM is employed to create two surrogate functions for the RCI network, and then PGD is used to solve for this problem.  As a result, Algorithm \ref{alg1} can obtain an optimal array of resource allocation in different radar and communication integration networks.
	
	\begin{center}
		\includegraphics[width=.8\linewidth]{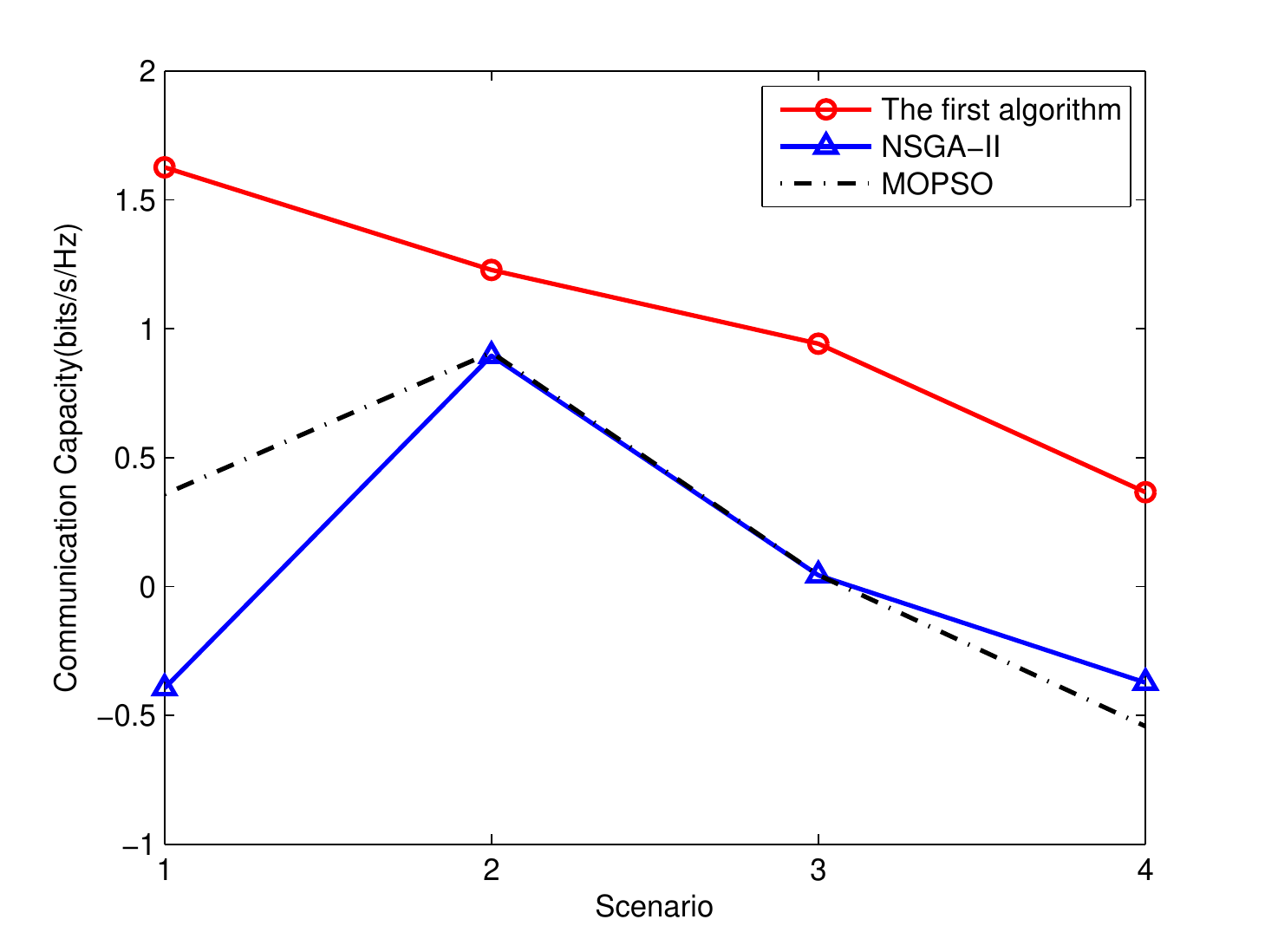}
		\captionof{figure}{  Average Communication Capacity comparison in different scenarios}
		\label{fig7}
	\end{center}
	
	\begin{center}
		\includegraphics[width=.8\linewidth]{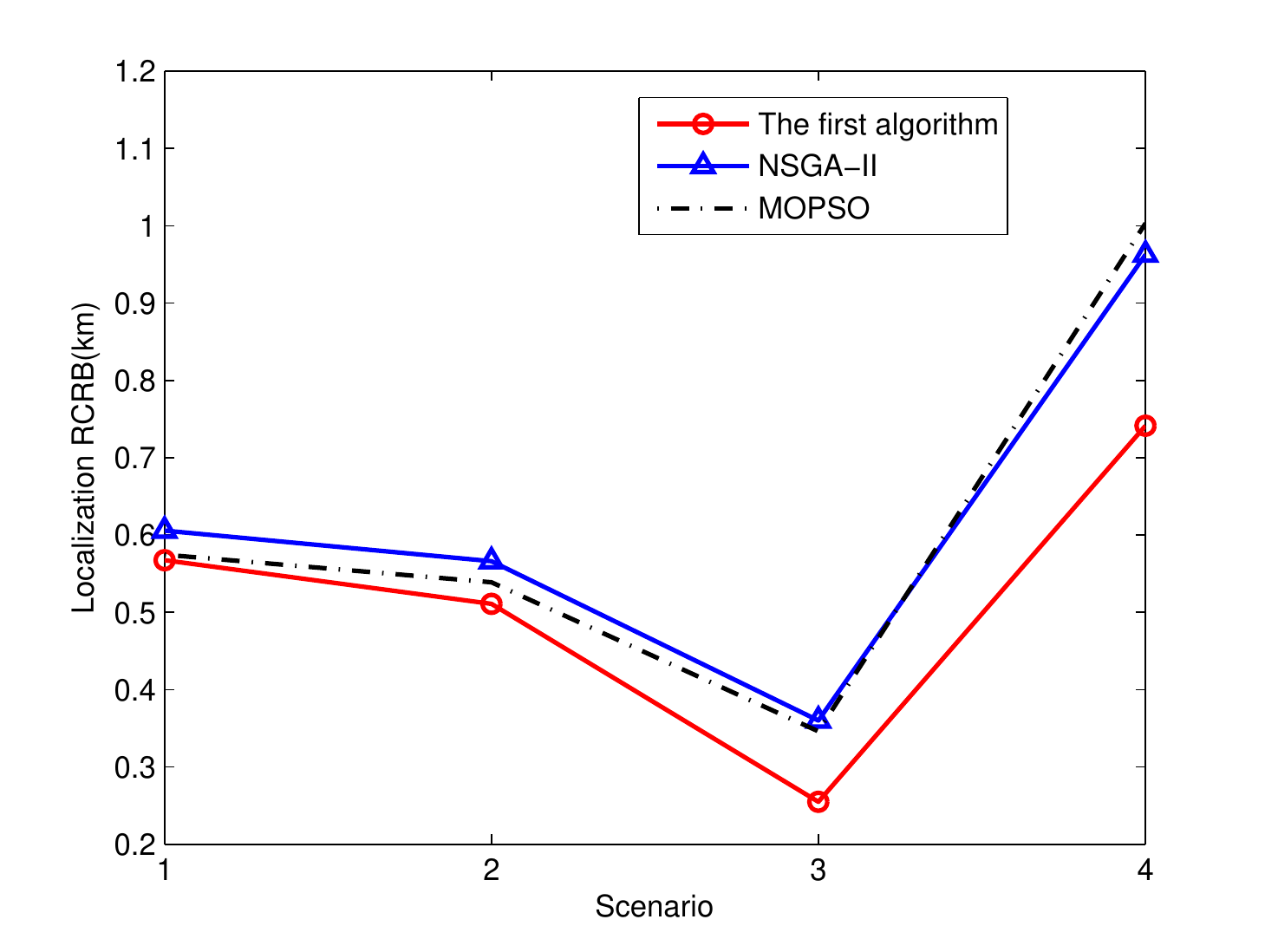}
		\captionof{figure}{  Average Localization performance comparison in different scenarios}
		\label{fig8}
	\end{center}

	In the third simulation, Algorithm \ref{alg2} is compared with the traditional Genetic algorithm (GA)  optimization \cite{GA2019}, and Particle swarm optimization(PSO)\cite{PSO2020} methods, which seem to be good at solving the single-objective optimization problems.
	Ten scenarios with three targets and three platforms are considered here.
	The average localization RCRB of each target for all scenarios is shown in Fig. \ref{fig9}, when the communication capacity threshold is set at 5 bits/s/Hz.  As expected, Algorithm \ref{alg2} can obtain a lower RCRB than the GA and PSO based methods.
	In one of the scenarios, localization performance is evaluated with different communication capacity thresholds from 1 bit/s/Hz to 5 bits/s/Hz, as illustrated in Fig. \ref{fig10}. The communication tasks need much more array resources due to the increase of capacity threshold. The localization performance becomes poorer with the decrease of the array resource used for radar localization.

	\begin{center}
		\includegraphics[width=.8\linewidth]{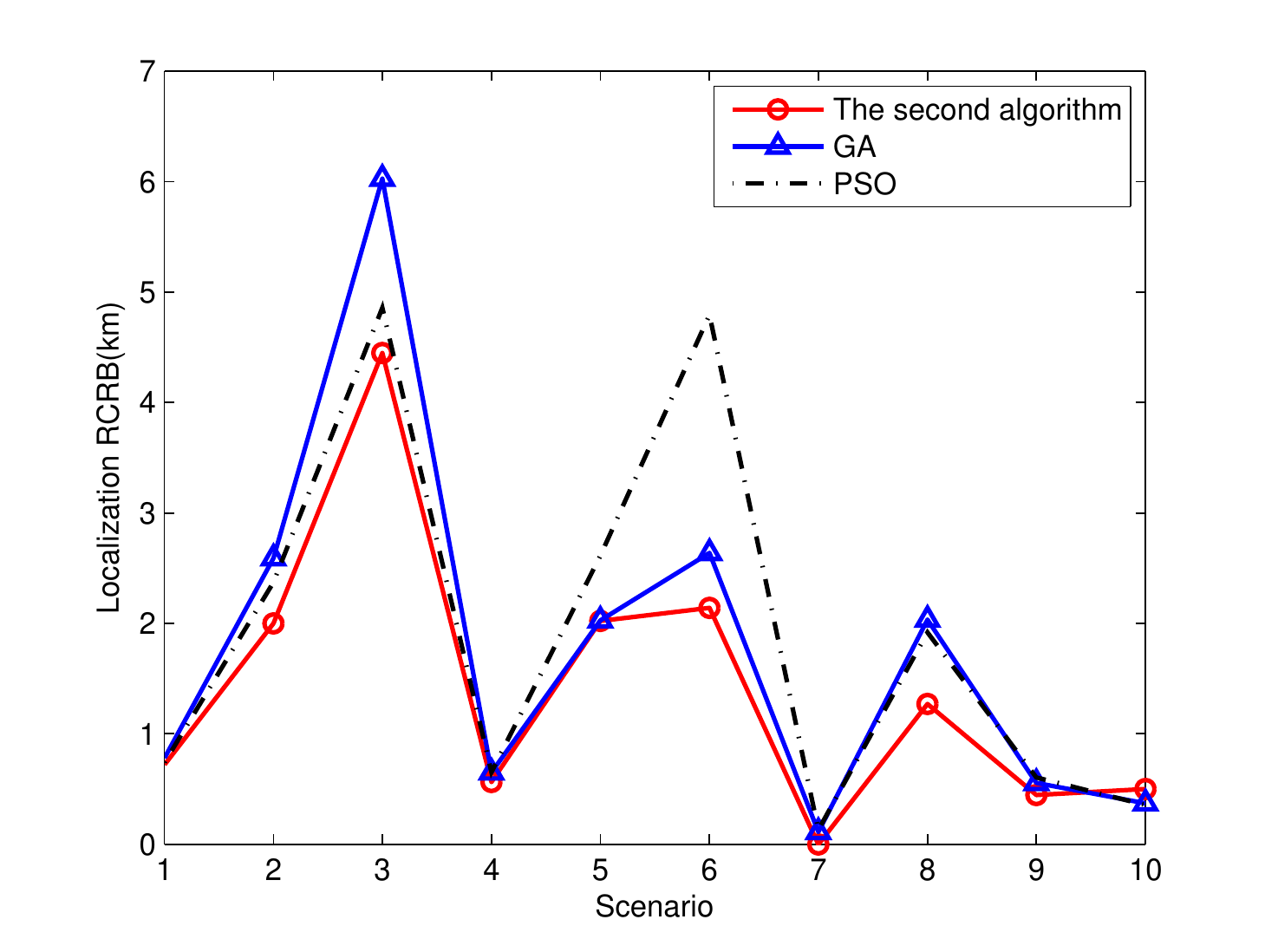}
		\captionof{figure}{  Average Localization performance comparison(3 targets and 3 platforms)}
		\label{fig9}
	\end{center}

	\begin{center}
		\includegraphics[width=.8\linewidth]{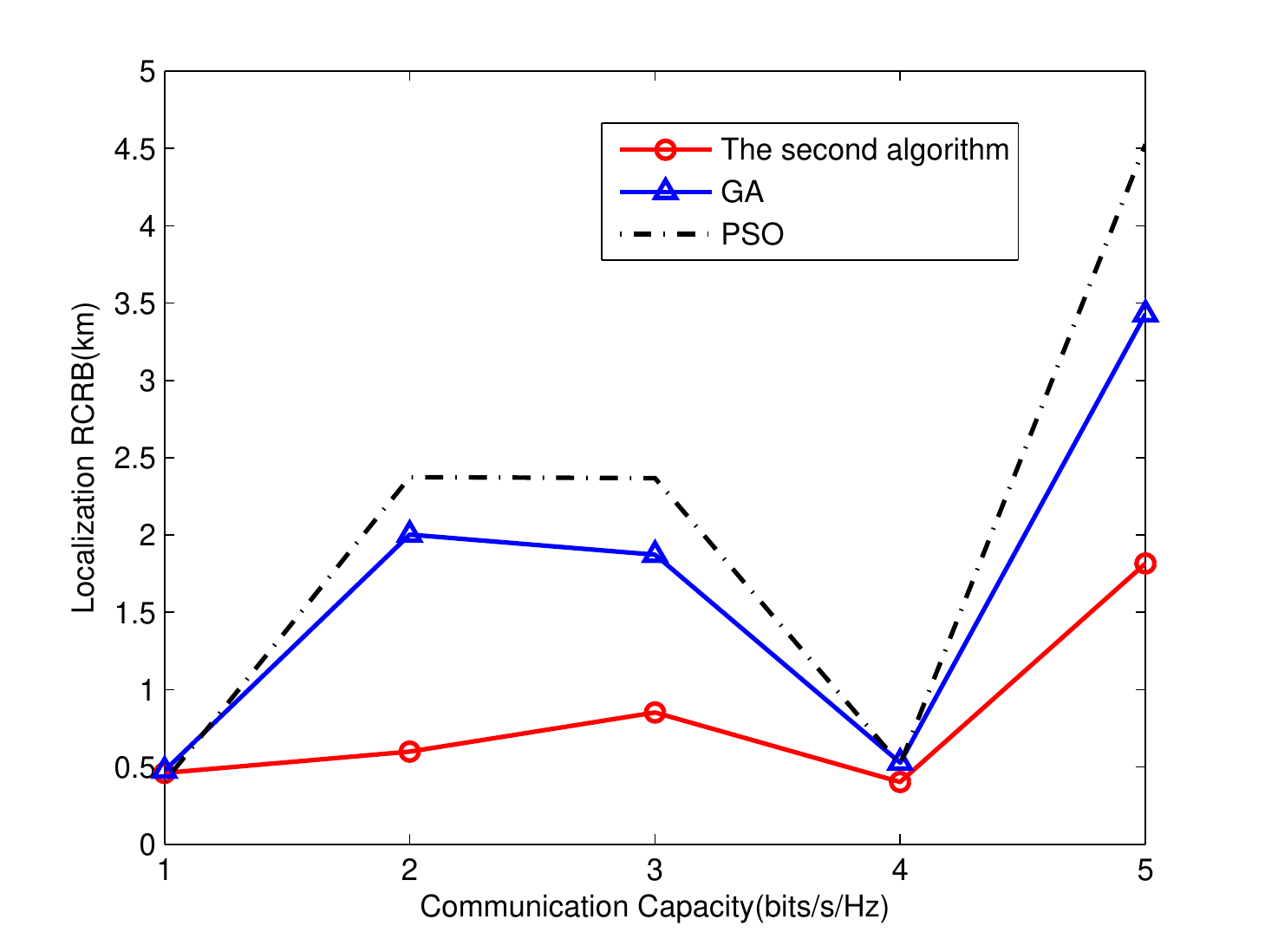}
		\captionof{figure}{  Average Localization performance comparison (different communication capacity threshold)}
		\label{fig10}
	\end{center}
	
	Consequently, the same four simulation scenarios for Algorithm \ref{alg1} are considered. Fig. \ref{fig11} depicts the average localization performance, under the same communication capacity constraint. It is clear that the proposed method has better localization performance, especially for the third and fourth scenarios, in which there are more targets and platforms.
	Fig. \ref{fig12} displays the localization performance in the environment of 12 targets and 12 platforms. Once again, we see that Algorithm \ref{alg2} outperforms the GA and PSO methods in allocating arrays for the RCI network.
	\begin{center}
		\includegraphics[width=.8\linewidth]{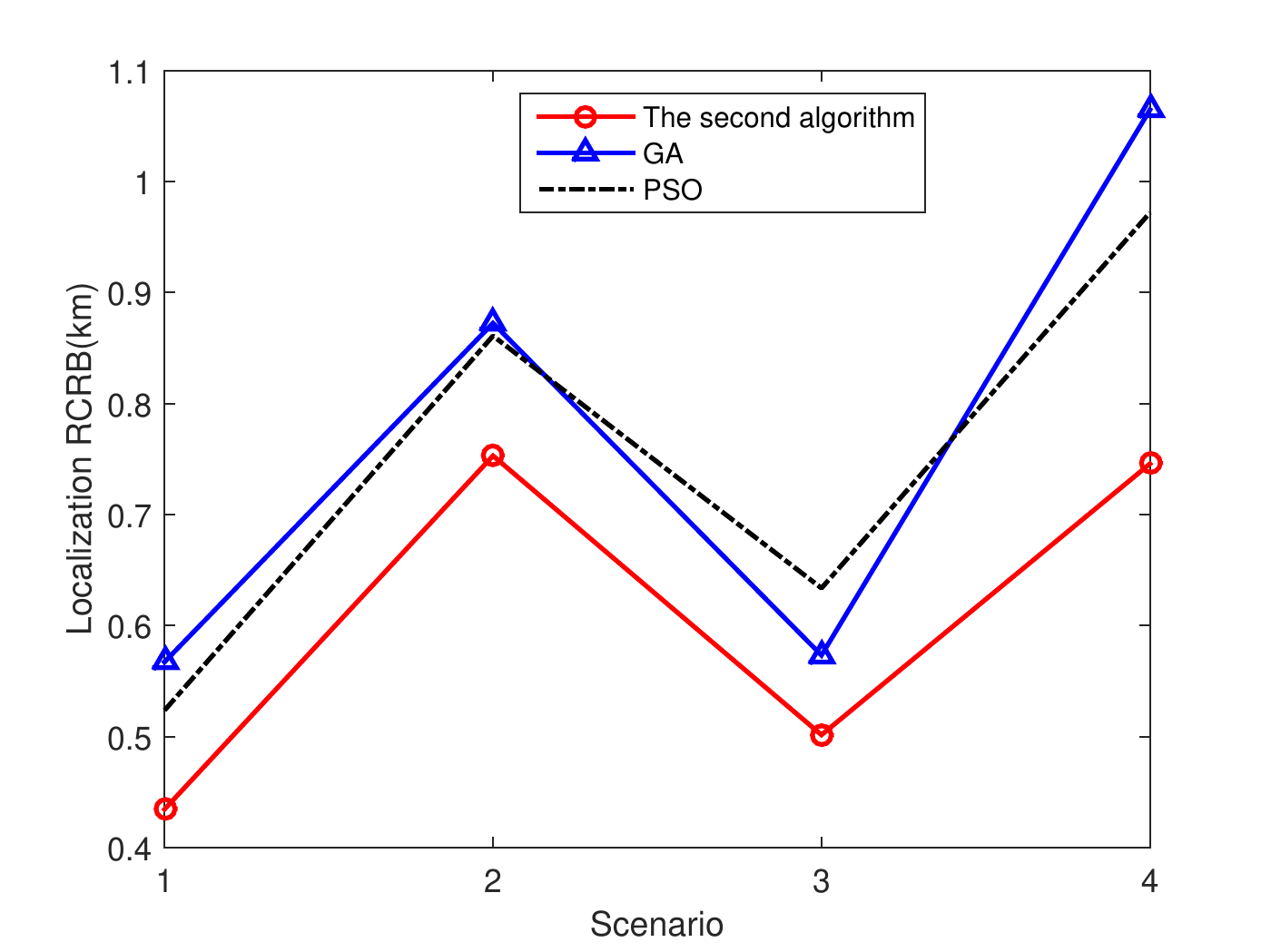}
		\captionof{figure}{  Average Localization performance comparison in different scenarios}
		\label{fig11}
	\end{center}
	
	\begin{center}
		\includegraphics[width=.8\linewidth]{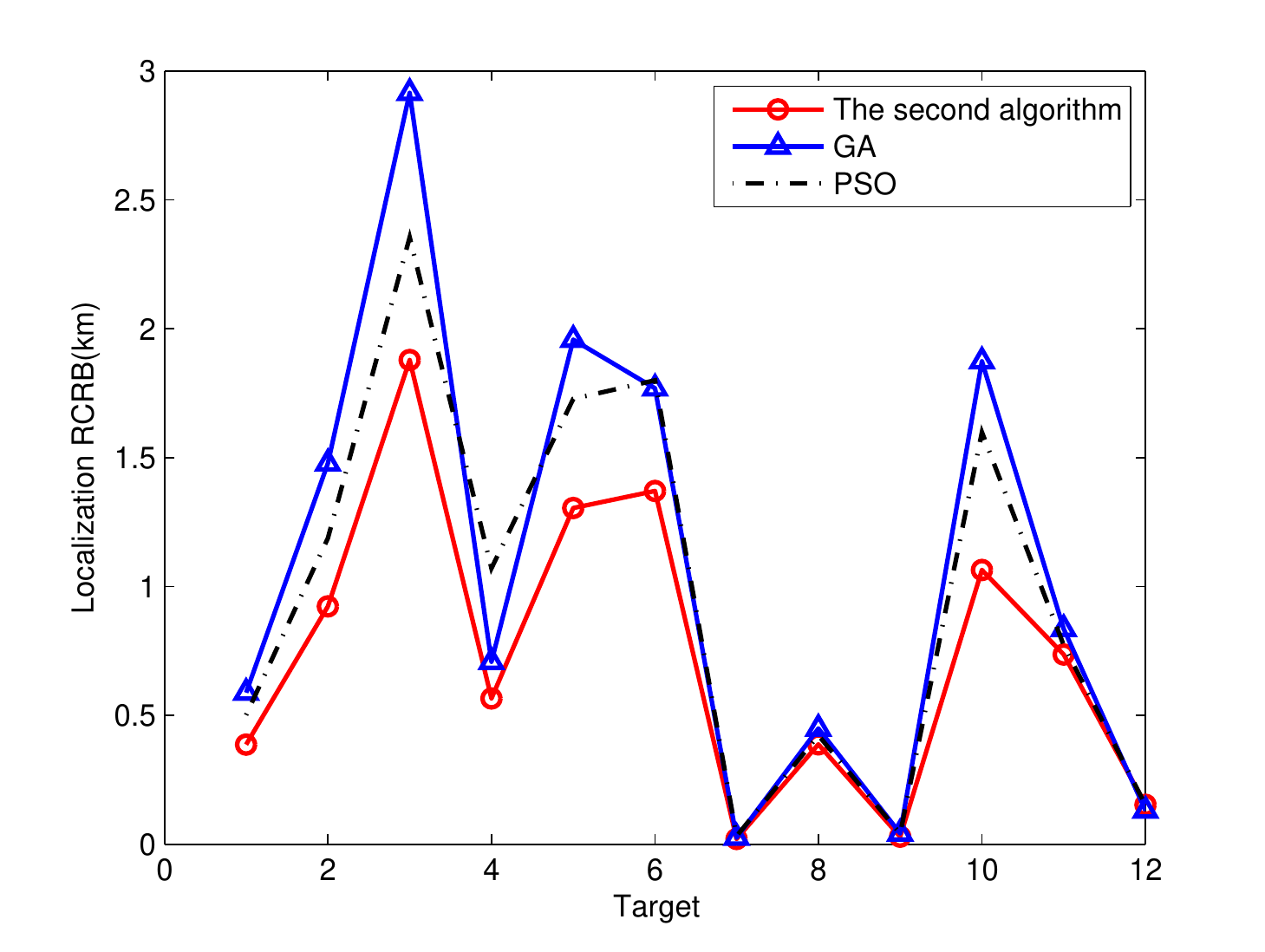}
		\captionof{figure}{  Localization performance comparison of every target}
		\label{fig12}
	\end{center}

	\section{Conclusion}\label{sec5}
	We proposed a joint method of MM and PGD to perform array resource allocation in an RCI network. The RCI model, including localization and communication function, is used. To further simplify the optimization problem, the MM method is employed to design the surrogate functions for array allocation. With and without communication capacity constraint, the PGD method is used to solve the two optimization problems. Our simulation results show that the two proposed algorithms have improved performance compared to the classic optimization methods.

	\bibliography{reference}
	
\end{document}